\documentclass[nonacm]{acmart}
\usepackage{algorithm, pseudocode}
\usepackage{booktabs} 
\usepackage{graphicx}
\usepackage{amsfonts}
\usepackage{amsmath}
\usepackage{amsthm}
\usepackage{bm}
\usepackage{etex}
\usepackage{color}
\usepackage{graphicx}
\usepackage{esint}
\usepackage{tabularx}
\usepackage{multirow}
\usepackage[tight]{subfigure}
\usepackage{multirow}
\usepackage{mathtools}
\usepackage[noend]{algpseudocode}
\usepackage[colorinlistoftodos]{todonotes}
\usepackage{chngcntr}
\usepackage{titlesec}
\usepackage{stfloats}
\usepackage{adjustbox}

\PassOptionsToPackage{table,xcdraw}{xcolor}

\usepackage{graphicx} % For including images
\usepackage{tikz-cd}
\usepackage{mathtools}
\usepackage{placeins}

\newtheorem{definition}{Definition}
\newtheorem{theorem}{Theorem}
\newtheorem{lemma}{Lemma}
\newtheorem{remark}{Remark}

\newtheorem{example}{Example}

\newcommand{\A}{\mathbb{A}}
\newcommand{\Q}{\mathbb{Q}}
\renewcommand{\P}{\mathbb{P}}

\newcommand{\F}{\mathbb{F}}
\newcommand{\Z}{\mathbb{Z}}
\newcommand{\N}{\mathbb{N}}
\renewcommand{\L}{\mathbb{L}}
\newcommand{\ang}[1]{\{#1\}}

\newcommand{\frakp}{\mathfrak{p}}

\newcommand{\softO}{O\tilde{~}}

\newcommand{\sring}{\mathbb{L}\{ \tau \}}

\newcommand{\charpoly}{\textnormal{CharPoly}}
\newcommand{\minpoly}{\textnormal{MinPoly}}
\newcommand{\cA}{\mathcal{A}} \newcommand{\cC}{\mathcal{C}}
\newcommand{\cB}{\mathcal{B}}
\newcommand{\ehat}{\hat{\eta}} \newcommand{\emorph}{\textnormal{End}}

\newcommand{\frakl}{\mathfrak{l}}
\newcommand{\WL}{W}
\newcommand{\hcrys}{H^{*}_{{\rm crys}}(\phi, \L)}
\newcommand{\hcrysk}{H_{k}^*(\phi, \L)} 

\newcommand{\var}{y}
\newcommand{\vart}{t}

\begin{document}
%%% ACM Formatting %%%
%\setcopyright{rightsretained}
%\setcopyright{acmlicensed}

\copyrightyear{2023}
\acmYear{2023}
\acmConference[ISSAC]{International Symposium on Symbolic and Algebraic Computation}{2023}{}
%\acmYear{2019}

\title{Computing the Characteristic
  Polynomial of Endomorphisms of a finite Drinfeld Module using
  Crystalline Cohomology}

\author{Yossef Musleh}
\affiliation{\institution{Cheriton School of Computer Science \\ University of Waterloo}
  \city{Waterloo} \state{Ontario} \country{Canada}}
\email{ymusleh@uwaterloo.ca}

\author{\'Eric Schost}
\affiliation{\institution{Cheriton School of  Computer Science \\ University of Waterloo}
  \city{Waterloo} \state{Ontario} \country{Canada}}
\email{eschost@uwaterloo.ca}

\begin{abstract}
  We present a new algorithm for computing the characteristic
  polynomial of an arbitrary endomorphism of a finite Drinfeld module
  using its associated crystalline cohomology. Our approach takes
  inspiration from Kedlaya's $p$-adic algorithm for computing the
  characteristic polynomial of the Frobenius endomorphism on a
  hyperelliptic curve using Monsky-Washnitzer cohomology. The method
  is specialized using a baby-step giant-step algorithm for the
  particular case of the Frobenius endomorphism, and in this case we
  include a complexity analysis that demonstrates asymptotic gains
  over previously existing approaches. 
\end{abstract}




%% \setcounter{secnumdepth}{4}
%% \renewcommand{\theparagraph}{\thesubsection.\arabic{paragraph}}
%% \counterwithin{paragraph}{subsection} % makes paragraph depend on
%% subsection
%% \titleformat{\paragraph}[runin]{\normalfont\normalsize\bfseries}{\theparagraph.}{1em}{}
%% \titlespacing*{\paragraph}{0em}{1ex}{1em}
%% \newcommand{\pref}[1]{{\bf\ref{#1}}}

\maketitle

%%%%%%%%%%%%%%%%%%%%%%%%%%%%%%%%%%%%%%%%%%%%%%%%%%%%%%%%%%%%
%%%%%%%%%%%%%%%%%%%%%%%%%%%%%%%%%%%%%%%%%%%%%%%%%%%%%%%%%%%%
%%%%%%%%%%%%%%%%%%%%%%%%%%%%%%%%%%%%%%%%%%%%%%%%%%%%%%%%%%%%

\section{Introduction}
Drinfeld modules were first introduced by Vladimir Drinfel'd in order
to prove the Langlands conjecture for $\textnormal{GL}_n$ over a
global function field~\cite{Drinfeld74}. Since then, Drinfeld modules
have attracted attention due to the well established correspondence
between elliptic curves and the rank two case. Moreover, the rank one
case, often referred to as \textit{Carlitz modules}, provides a
function field analogy of cyclotomic extensions; the role played in
class field theory over number fields by elliptic curves with complex
multiplication shows strong parallels with that of Drinfeld modules of
rank two for the function field setting.  This has motivated efforts
to translate constructions and algorithms for elliptic curves,
including modular polynomials~\cite{Caranay2020ComputingMP},
isogenies~\cite{Caranay2020ComputingMP}, and endomorphism
rings~\cite{endomorphismpink,garaipap}.

Naturally, cryptographic applications of Drinfeld modules have also
been explored~\cite{crsaction}, but were long anticipated to be
vulnerable for public key cryptography based on
isogenies~\cite{Scanlon01,Joux2019DrinfeldMA}. This question was
finally put to rest by Wesolowski who showed that isogenies between
Drinfeld modules of any rank could be computed in polynomial
time~\cite{wesolowski}.

Drinfeld modules of rank $r > 2$ do not have such a clear parallel,
although an analogy exists between abelian surfaces and so called
$t$-modules~\cite{Anderson86}. Owing to this discrepancy, rank two
Drinfeld modules have been studied far more closely than the case of
more general ranks.

The main goal of this work is to study a Drinfeld module analogue of
$p$-adic techniques such as Kedlaya's
algorithm~\cite{kedlayapointcounting} for computing the
characteristic polynomial of the Frobenius endomorphism acting on an
elliptic or hyperelliptic curve over a finite field. Algorithms for
elliptic curves compute the action of the Frobenius on a basis of a
particular subspace of the de Rham cohomology of a characteristic 0
lift of the curve, with coefficients in $\Q_p$. Our approach follows a
very similar outline, but turns out to be remarkably simpler to
describe, resting crucially on a suitable version of crystalline
cohomology for Drinfeld modules due Gekeler and
Angl\`es~\cite{Angles1997}.

More generally, the approach we present can be used to compute the
characteristic polynomial of any Drinfeld module endomorphism.

%%%%%%%%%%%%%%%%%%%%%%%%%%%%%%%%%%%%%%%%%%%%%%%%%%
%%%%%%%%%%%%%%%%%%%%%%%%%%%%%%%%%%%%%%%%%%%%%%%%%%
%%%%%%%%%%%%%%%%%%%%%%%%%%%%%%%%%%%%%%%%%%%%%%%%%%

\section{Background and Main result}

%%%%%%%%%%%%%%%%%%%%%%%%%%%%%%%%%%%%%%%%%%%%%%%%%%

\subsection{Basic Preliminaries}

Let $R$ be any ring, $r \in R$, and $\sigma: R \to R'$ a ring
homomorphism. We will follow the notational convention that writes
$\sigma(r) = \sigma_r = r^{\sigma}$ throughout this work. If $R$ is a
polynomial ring and $\sigma$ acts on its coefficient ring,
$r^{\sigma}$ denotes coefficient-wise application.

Let $q$ be a prime power, and let $\F_q$ denote a finite field of
order $q$, fixed throughout. We also fix a field extension $\L$ of
$\F_q$ such that $[\L: \F_q] = n$. Explicitly, $\L$ is defined as $\L
= \F_q[t]/(\ell(t))$ for some degree $n$ irreducible $\ell(t) \in
\F_q[t]$, so elements of $\L$ are represented as polynomials in
$\F_q[t]$ of degree less than $n$. We will discuss below an
alternative representation, better suited for some computations.

%%%%%%%%%%%%%%%%%%%%%%%%%%%%%%%%%%%%%%%%%%%%%%%%%%

\subsection{Drinfeld Modules}

In general, Drinfeld modules can be defined over a ring $A$ consisting
of the functions of a projective curve over $\F_q$ that are regular
outside of a fixed place at infinity. For our purposes, we will
restrict ourselves to the consideration of Drinfeld modules defined
over the regular function ring of $\P^1 - \{\infty\}$; that is $A =
\F_q[x]$.

We fix a ring homomorphism $\gamma: A \to \L$ and let $\frakp \in A$
the monic irreducible generator of $\ker \gamma$. Then $\F_{\frakp} =
\F_q[x]/(\frakp)$ is isomorphic to a subfield of $\L$; we let
$m=\deg(\frakp)$, so that $m$ divides $n$. This gives us an
isomorphism $\L\simeq\F_q[x,\vart]/(\frakp(x), g(x, \vart)) $, with
$g$ monic of degree $n/m$ in $\vart$. It will on occasion be
convenient to switch from the representation of elements of $\L$ as
univariate polynomials in $\vart$ to the corresponding bivariate
representation in $x,\vart$; in that case, for instance, $\gamma_x$ is
simply the residue class of $x$ modulo $(\frakp(x), g(x, \vart))$. We
assume that $\frakp$ and $g$ are given as part of the input.
  
To define Drinfeld modules, we also have to introduce the ring
$\L\ang{\tau}$ of skew polynomials, namely
\begin{align*}
\L\ang{\tau} &= \{U=u_0 + u_1 \tau + \cdots + u_s \tau^s \ \mid \ s
\in \N, u_0,\dots,u_s \in \L\},
\end{align*}
where multiplication is induced by the relation $\tau u = u^q \tau$,
for all $u$ in $\L$.

\begin{definition}
 A \textnormal{Drinfeld $A$-module} of rank r over over $(\L, \gamma)$
 is a ring homomorphism $\phi: A \to \L\{ \tau \}$ such that
 \begin{equation*}
     \phi_x = \gamma_x + \Delta_1\tau^1 + \ldots + \Delta_r\tau^r
 \end{equation*}
  with $\Delta_i$ in $\L$ for all $i$ and $\Delta_r \neq 0$.
 \end{definition}

For readers interested in the more general setting under which
Drinfeld modules are typically defined, we recommend the survey by
Deligne and Husem\"{o}ller in~\cite{deligne}.

A Drinfeld module is defined over the \textit{prime field} when $\L
\cong \F_{\frakp}$ (that is, $m=n$). Algorithms for Drinfeld modules
in the prime field case tend to be algorithmically simpler, and we
will often highlight the distinction with the more general case.

 \begin{example}
     Let $\F_q = \Z/5\Z$ and $n = 4$. Set $\ell(t) = t^4 + tx^2 + 4t +
     2$ and $\L = \F_5[t]/(\ell(t))$. Let $\gamma_x = t \bmod
     \ell(t)$, in which case $\L = \F_{\frakp}$. A rank two Drinfeld
     module is given by $\phi_x = \tau^2 + \tau + t$.

     We may instead take $\gamma_x = t^3 + t^2 + t + 3 \bmod \ell(t)$
     in which case $\frakp = x^2 + 4x + 2$ and $\F_{\frakp} \cong
     \F_{25}$. The field $\L$ admits the representations
     $$\L=\F_5[t]/(\ell(t)) \simeq \F_5[x,t]/(\frakp(x),g(x,t)),$$
     with $g(x,t) = t^2 + 4tx + 3t + x$. A rank three Drinfeld module
     is given by $\phi_x = \tau^3 + (t^3 + 1)\tau^2 + t \tau + t^3 +
     t^2 + t + 3$.
 \end{example}
 
Given Drinfeld $A$-modules $\phi, \psi$ defined over $(\L, \gamma)$,
an $\L$-morphism $u: \phi \to \psi$ is a $u \in \L\{ \tau \}$ such
that $u\phi_a = \psi_au $ for all $a \in A$. The set
$\emorph_{\L}(\phi)$ is the set of $\L$-morphisms $\phi \to \phi$; it
is therefore the centralizer of $\phi_x$ in $\L\{\tau\}$.  It
admits a natural ring structure, and  contains the
\textit{Frobenius endomorphism} $\tau^n$.

%%%%%%%%%%%%%%%%%%%%%%%%%%%%%%%%%%%%%%%%%%%%%%%%%%%%%%%%%%%%

\subsection{Characteristic Polynomials}\label{ssec:charpoly}

The characteristic polynomial of an endomorphism $u \in
\emorph_{\L}(\phi)$ can be defined through several points of view.

Through the action of $\phi$, $A=\F_q[x]$ and its fraction field
$K=\F_q(x)$ can be seen as a subring, resp. subfield of the skew field
of fractions $\L(\tau)$ of $\L\{\tau\}$. Then, $\emorph_{\L}^0(\phi)=
\emorph_{\L}(\phi) \otimes_{A} K$ is the centralizer of $\phi_x$ in
$\L(\tau)$; this is a division ring that contains $K$ in its center.

\begin{definition}
  The \textit{characteristic polynomial} $\charpoly(u)$ of $u \in
  \emorph_{\L}(\phi)$ is its reduced characteristic polynomial,
  relative to the subfield $K$ of
  $\emorph_{\L}^0(\phi)$~\cite[Section~9.13]{Reiner03}.
\end{definition}

The characteristic polynomial of $u$ has degree $r$ and coefficients
in $A \subset K$, so that it belongs to $A[Z]$. More precisely, if
$\deg(u) = d$, $\charpoly(u)$ has coefficients $a_0, \ldots, a_{r-1}
\in A$ with $\deg(a_i) \leq {d(r - i)}/r$ for all
$i$~\cite[Prop.~4.3]{endomorphismpink} and satisfies
\begin{equation}
    u^r + \sum_{i=0}^{r-1} \phi_{a_i}u^i = 0.
\end{equation}

Another definition of $\charpoly(u)$ follows from the introduction of
the {\em Tate modules} of $\phi$. The Drinfeld module $\phi$ induces
an $A$-module structure on the algebraic closure $\overline{\L}$ of
$\L$ by setting $a * c = \phi_a(c)$ for $a \in A$, $c \in
\overline{\L}$ (defining $\tau^i(c)=c^{q^i}$). Then, for $\frakl \in
A$, the $\frakl$-torsion module of $\phi$ is defined as $\phi[\frakl]
= \{ c \in \overline{\L} \mid \frakl *c = 0 \}$. Setting $\frakl$ to
be any irreducible element of $A$ different from $\frakp$, we can
define the $\frakl$-adic Tate module as $T_{\frakl}(\phi) =\varprojlim
\phi[\frakl^i]$.

Letting $A_{\frakl}$ be the $\frakl$-adic completion of $A$,
$T_{\frakl}(\phi)$ becomes a free $A_{\frakl}$-module of rank $r$ and
elements of $\emorph_{\L}(\phi)$ induce endomorphisms on
$T_{\frakl}(\phi)$. Then, for $u \in \emorph_{\L}(\phi)$, the
characteristic polynomial $\charpoly_{A_{\frakl}}(u)$ of the induced
endomorphism $u \in \emorph_{A_{\frakl}}(T_{\frakl}(\phi))$ agrees
with $\charpoly(u)$~\cite{Gekeler91,Angles1997}.

\begin{example}
    Let $\F_q$, $\L$ be as in the context of example 1, and $\gamma_x
    = t^3 + 4t^2 + t + 1 \bmod \ell(t)$.
    A rank 5 Drinfeld module is given by $\phi_x = (4t^3 + t^2 +
    2)\tau^5 + (t^3 + 3t^2 + t + 1)\tau^4 + (4t + 3)\tau^3 + (3t^2 +
    4t + 4)\tau^2 + (4t^3 + 4t^2 + 4t)\tau + \gamma_x$.

    The characteristic polynomial of $\tau^n$ on $\phi$ is $Z^5 + 3Z^4
    + (x^3 + 4x^2 + x)Z^3 + (2x^2 + 4x + 3)Z^2 + (x^3 + 2x^2 + 4x +
    2)Z$ $ + 2x^4 + 3x^2 + 4x + 2$
\end{example}

The results in this paper are based on another interpretation of
$\charpoly(u)$, as the characteristic polynomial of the endomorphism
induced by $u$ in a certain {\em crystalline cohomology} module, due
to Gekeler and Angl\`es~\cite{Angles1997}. Our first main result is an
algorithm for computing the characteristic polynomial of the Frobenius
endomorphism.

Here, $\omega$ is a real number such that two $s \times s$ matrices
over a ring $R$ can be multiplied in $O(s^{\omega})$ ring operations
in $R$; the current best value is $\omega \leq 2.372$~\cite{DuWuZh22}.
We will also let $\lambda$ denote an exponent such that the
characteristic polynomial of an $s \times s$ matrix over a ring $R$
can be computed in $O(s^{\lambda})$ ring operations in $R$. When $R$
is a field, this can be done at the cost of matrix multiplication and
therefore $\lambda = \omega$~\cite{charpolycomp}. For more general
rings, the best known value to date is $\lambda \approx
2.7$~\cite{KaVi04}.

\begin{theorem}\label{mainresult}
  Let $\phi$ be a rank $r$ Drinfeld module over $(\L, \gamma)$. There
  is a deterministic algorithm to compute the characteristic
  polynomial of the Frobenius endomorphism $\tau^n$ with bit
  complexity
  \begin{itemize} 
      \item $(r^\omega n^{1.5} \log q + n \log^2 q)^{1+o(1)}$ for the prime field
        case ($m = n$)
      \item $((r^{\lambda}/m + r^\omega/\sqrt{m})n^2\log q + n \log^2 q)^{1+o(1)}$ for
        the general case $m < n$.
  \end{itemize}
\end{theorem}
When $r$ and $q$ are fixed, the runtime in the theorem is thus
essentially linear in $n^2/\sqrt{m}$, which is $n^{1.5}$ in the prime
field case and gets progressively closer to $n^2$ as $m$
decreases. The best prior results~\cite{MuslehSchost} were limited to
the case $r=2$, with runtimes essentially linear in $n^{1.5}$ in the
prime field case and $n^2$ otherwise (for fixed $q$).

This first algorithm builds upon techniques for linear recurrences
originating from~\cite{DOLISKANI2021199}, which are so far limited to
the particular case of the Frobenius endomorphism.

We also obtain two algorithms that can be applied to any $u \in
\emorph_{\L}(\phi)$. The complexity in this case partly depends on
that of multiplication and Euclidean division in $\L\{\tau\}$, which
we will denote ${\sf SM}(d,n,q)$ and which will be discussed in more
detail in Section~\ref{sec:prelim}.

\begin{theorem}\label{mainresult2}
  With assumptions as in Theorem \ref{mainresult}, there are
  deterministic algorithms to compute the characteristic polynomial of
  an endomorphism $u$ of degree $d$ with bit complexities
  \begin{itemize}
      \item $\big( \frac{\min(dr^2, (d+r)r^{\omega - 1})}{m} (d+m)n \log q + r^{\lambda} n(d+m)/m\log q + n \log^2 q \big)^{1+o(1)}$
      \item $(r{\sf SM}(d+r,n, q) + r^{\lambda} n(d+m)/m\log q + n \log^2 q)^{1+o(1)} $.
  \end{itemize}
\end{theorem}
Again, it is worth considering the situation with $r$ and $q$
fixed. In this case, the runtimes we obtain are, respectively,
essentially linear in $d(d+m)n/m$ and ${\sf SM}(d,n, q)$. In the next
section, we review known values for ${\sf SM}$; for the best known
value of $\omega$, and fixed $q$, it is $(d^{1.69} n)^{1+o(1)}$ for $d
\le n^{0.76}$, and $(d n^{1.52})^{1+o(1)}$ otherwise. In the case
$d=\Theta(n)$, the runtimes are thus essentially linear in $n^3/m$ and
$n^{2.53}$, respectively (so which is the better algorithm depends on
the value of $m$). For $u=\tau^n$, the algorithm in the previous
theorem is of course superior.

%%%%%%%%%%%%%%%%%%%%%%%%%%%%%%%%%%%%%%%%%%%%%%%%%%%%%%%%%%%%
%%%%%%%%%%%%%%%%%%%%%%%%%%%%%%%%%%%%%%%%%%%%%%%%%%%%%%%%%%%%
%%%%%%%%%%%%%%%%%%%%%%%%%%%%%%%%%%%%%%%%%%%%%%%%%%%%%%%%%%%%

\section{Computational Preliminaries}\label{sec:prelim}

The key element in our complexity analyses is the cost of the
following operations in $\L$: addition/subtraction, multiplication,
inverse and (iterated) Frobenius application.

Some of the algorithms we use below (multiplication and Euclidean
division in $\L\{\tau\}$ from~\cite{PUCHINGER2017b,CaLe17}) assume
that all these operations can be done using $\softO(n)$ operations in
$\F_q$. For the representation of $\L$ we use, this is however not
known to be the case; Couveignes and Lercier proved the existence of
``elliptic bases'' that satisfy these requirements~\cite{CoLe09}, but
conversion to our representation does not appear to be obvious.

This explains why in our main result, we do not count operations in
$\F_q$, but bit operations instead (our complexity model is a standard
RAM); we explain below how this allows us to bypass the constraints
above.

Using FFT based algorithms, polynomials of degree at most $n$ with
coefficients in $\F_q$ can be multiplied in boolean time $(n \log
q)^{1 + o(1)}$ \cite{Cantor1991OnFM,HaHoLe17}. It follows that
elementary field operations (addition, multiplication, inversion) in
$\L=\F_q[t]/(\ell(t))$ can be done with the same asymptotic cost.

Conversions between univariate and bivariate representations for
elements of $\L$ take the same asymptotic runtime.
Denote by $\alpha$ the isomorphism $\L=\F_q[t]/(\ell(t)) \to
\F_q[x,t]/(\frakp(x),g(x,t))$; then, given $f$ of degree less than $n$
in $\F_q[t]$, we can compute the image $\alpha(f \bmod \ell(t))$ in $(n
\log q)^{1 + o(1)}$ bit operations; the same holds for
$\alpha^{-1}$~\cite{PoSc13a,VANDERHOEVEN2020101405}.

%% For most standard computations outside of modular composition,
%% conversion from a bit complexity runtime to an algebraic complexity
%% can be obtained by dividing by $ (n \log q)^{1 + o(1)}$.

The last important operation is the application of the $q$-power
Frobenius in $\L$. Recall that given polynomials $f,g, h \in \F_q[x]$
of degree at most $n$, {\em modular composition} is the operation that
computes $f(g) \bmod h$. As showed in~\cite{vonzurGathen1992}, for $c$
in $\L=\F_q[t]/(\ell(t))$, $c^q$ can be computed in the same
asymptotic time (up to logarithmic factors) as degree $n$ modular
composition, following a one-time precomputation that takes $(n \log^2
q)^{1+o(1)}$ bit operations. This then extends to arbitrary powers
(positive and negative) of the Frobenius.  We
should point out that modular composition techniques also underlie the
algorithms for switching between the two representations of the
elements in $\L$ mentioned above.

In~\cite{KedUman}, Kedlaya and Umans proved that modular composition
in degree $n$ can be computed in $(n \log q)^{1 + o(1)}$ bit
operations (see also the refinement due to van der Hoeven and
Lecerf~\cite{VANDERHOEVEN2020101405}), whence a similar cost for
(iterated) Frobenius in $\L$. Here, the fact that we work in a boolean
model is crucial: Kedlaya and Umans' algorithm is not known to admit a
description in terms of $\F_q$-operations.

From this, we can directly adapt the cost analyses in
\cite{PUCHINGER2017b,CaLe17} to our boolean model. In particular,
following the latter reference (which did so in an algebraic cost
model), we let ${\sf SM}(d,n,q)$ be a function such that
\begin{itemize}
\item degree $d$ multiplication and right Euclidean division in
  $\L\{\tau\}$ can be done in $O({\sf SM}(d,n,q))$ bit operations
\item for $n,q$ fixed, $d \mapsto {\sf SM}(d,n,q)/d$ is non-decreasing.
\end{itemize}
The latter condition is similar to the super-linearity of
multiplication functions used in~\cite{GaGe13}, and will allow us to
streamline some cost analyses. Unfortunately, there is no simple
expression for ${\sf SM}(d,n,q)$: on the basis of the algorithms
in~\cite{PUCHINGER2017b,CaLe17}, the analysis done in~\cite{CaLe17}
gives the following upper bounds:
\begin{itemize}
\item for $d \le n^{(5-\omega)/2}$, we can take
  ${\sf SM}(d,n,q)$ in $(d^{(\omega+1)/2} n \log q)^{1+o(1)}$
\item else, we can take ${\sf SM}(d,n,q)$ in $(d n^{4/(5-\omega)} \log
  q)^{1+o(1)}$
\end{itemize}
For instance, with $d=n$, this is $(n^{(9-\omega)/(5-\omega)} \log q
)^{1+o(1)}$.

With $\omega=2.37$, the cost is $(d^{1.69} n \log
q)^{1+o(1)}$ for $d \le n^{0.76}$, and $(d n^{1.52}\log q)^{1+o(1)}$
otherwise; the exponent for $d=n$ is $2.53$. For completeness, we
point out that these algorithms heavily rely on Frobenius
applications, and as such, require spending the one-time cost $(n
\log^2 q)^{1+o(1)}$ mentioned previously.

% If $k = 1$, then $W_k = \L$ and fast techniques for computing the
% characteristic polynomial of a matrix over a field can be leveraged
% for a bit complexity of $(r^{\omega}n \log q)^{1 +
%   o(1)}$~\cite{charpolycomp}. For all other $k$, this may be done in
% $(r^{2.7} kn\log q)^{1+o(1)}$ bit operations, which is also
% $((d+m)r^{2.7} n/m\log q)^{1+o(1)}$, using the algorithm
% of~\cite{KaVi04}. For $d=n$, the latter cost is thus $(r^{2.7}
% n^2/m\log q)^{1+o(1)}$.

One should also keep in mind that these asymptotic cost analyses are
not expected to reflect practical runtimes. To the authors' knowledge,
software implementations of the Kedlaya-Umans algorithm achieving its
theoretical complexity, or of matrix multiplication with exponent
close to $2.37$, do not currently exist. For practical purposes,
implementations of modular composition use an algorithm due to Brent
and Kung~\cite{BrKu78}, with an algebraic complexity of $O(n^{(\omega
  + 1)/2}) $ operations in $\F_q$. Revisiting skew polynomial
algorithms and their analyses on such a basis is work that remains to
be done.

%%%%%%%%%%%%%%%%%%%%%%%%%%%%%%%%%%%%%%%%%%%%%%%%%%%%%%%%%%%%
%%%%%%%%%%%%%%%%%%%%%%%%%%%%%%%%%%%%%%%%%%%%%%%%%%%%%%%%%%%%
%%%%%%%%%%%%%%%%%%%%%%%%%%%%%%%%%%%%%%%%%%%%%%%%%%%%%%%%%%%%

\section{Prior Work}

The question of computing the characteristic polynomial, particularly
of the Frobenius endomorphism, was studied in detail
in~\cite{gekfrobdist} for the rank two case only.

The most general approach constructs a linear system based on the
degree constraints of the coefficients $a_i = \sum_{j = 0}^{n(r-i)/r}
a_{i,j} x^j$. Evaluating the characteristic polynomial at the
Frobenius element and equating coefficients gives a linear system
based on
\begin{equation}
    \tau^{nr} + \displaystyle\sum_{i=0}^{r-1}\sum_{j =
      0}^{\frac{n(r-i)}{r}}\sum_{k = 0}^{n(r-i)} a_{i,j}f_{j,k}
    \tau^{k + ni} = 0,
\end{equation}
with $f_{j,k}$ the coefficients of $\phi_{x^{j}}$. Letting
$\minpoly(\tau^n)$ denote the minimal polynomial of $\tau^n$ (as an
element of the division algebra $\emorph_{\L}^0(\phi)$ over the field
$K=\F_q(x)$), the solution of the preceding system is unique and
yields the characteristic polynomial if and only if $\minpoly(\tau^n)
= \charpoly(\tau^n)$.

Garai and Papikian gave an algorithm for computing the characteristic
polynomial~\cite[\S 5.1]{garaipap} valid for the prime field case
only. As with the previous approach, this relies on the explicit
computation of $\phi_{x^i}$, which is the dominant computational
step. This can be done by $O(n^2)$ evaluations of the recurrence
\[f_{i+1,j} = \gamma_x^{q^j}f_{i,j} + \sum_{t=1}^{r} \Delta_t^{q^{j-
    t}} f_{i, j-t}.\]
Thus the bit complexity of computing all of
$\phi_x, \phi_{x^2}, \ldots, \phi_{x^n}$ is $(r n^3\log(q))^{1 +
  o(1)}$.

Further study of algorithms for the specific case of the Frobenius
endomorphism in rank $r=2$ was done in~\cite{Narayanan18}
and~\cite{MuslehSchost}. The latter focused on the complexity of the
algorithms and used the same computational model that will be used
here. As we reported after Theorem~\ref{mainresult}, the best known
runtime to date was quadratic in $n$, except in the case where
$\minpoly(\tau^n) = \charpoly(\tau^n)$, or in the prime field case
where a bit cost of $(n^{1.5} \log q + n \log^2 q)^{1+o(1)}$ is
possible \cite{DOLISKANI2021199}. To our knowledge, no previous
analysis is available for an arbitrary endomorphism $u$.

In the context of elliptic curves, Kedlaya's algorithm
\cite{kedlayapointcounting} computes the characteristic polynomial of
a matrix representation of the lift of the Frobenius action to a
subspace of the Monsky-Washnitzer cohomology, up to some finite
precision. Our algorithm follows the same high-level approach: we
compute a matrix for the endomorphism acting on the crystalline
cohomology with coefficients in a power series ring analogue to Witt
vectors. The induced endomorphism turns out to be quite simple to
describe in terms of skew-polynomial multiplication, which eliminates
the need for a complicated lifting step.

%%%%%%%%%%%%%%%%%%%%%%%%%%%%%%%%%%%%%%%%%%%%%%%%%%%%%%%%%%%%
%%%%%%%%%%%%%%%%%%%%%%%%%%%%%%%%%%%%%%%%%%%%%%%%%%%%%%%%%%%%
%%%%%%%%%%%%%%%%%%%%%%%%%%%%%%%%%%%%%%%%%%%%%%%%%%%%%%%%%%%%

\section{Crystalline Cohomology}

In this section, we first review the construction of the crystalline
cohomology of a Drinfeld module and its main properties; this can be
found in~\cite{Angles1997}, where the definition is credited to
unpublished work of Gekeler. Then, we introduce truncated versions of
these objects, which reduce the computation of characteristic
polynomials of endomorphisms of a Drinfeld module to characteristic
polynomial computations of matrices over truncated power series rings.

%%%%%%%%%%%%%%%%%%%%%%%%%%%%%%%%%%%%%%%%%%%%%%%%%%%%%%%%%%%%

\subsection{Definition}

The contents of this subsection is
from~\cite{Gekeler88,Angles1997}. The set of {\em derivations}
$D(\phi, \L)$ of a Drinfeld module $\phi$ is the set of $\F_q$-linear
maps $\eta: A \to \sring\tau$ satisfying the relation
\begin{equation*}
    \eta_{ab} = \gamma_a\eta_b + \eta_a\phi_b, \quad a,b \in A
\end{equation*}
Let then $\var$ be a new variable.  The set $D(\phi,\L)$ can be made
into an $\L[\var]$-module in the following manner.

% $\WL$ can be thought of as the space of Witt vectors for
% $\L[X]/I_L$; since $I_L$ is generated by a degree one polynomial the
% arithmetical structure of $\WL$ is trivial and this formalism is not
% required.
 
\begin{definition}{\cite[Section~2]{Angles1997}}
  The set $D(\phi, \L)$ is an $\L[\var]$-module under $ (c \var^i *
  \eta)_a =c \eta_a \phi_{x^i}$, for $\eta$ in $D(\phi, \L)$, $c$ in
  $\L$, $i \ge 0$ and $a$ in $A$.
\end{definition}

Let further $I$ be the ideal of $\L[\var]$ generated by $\var -
\gamma_x$; for $k \ge 1$, we set
\[W_k = \L[\var]/I^k$$ and $$\WL = \varprojlim \hspace{1mm} W_k \cong
\L[[\var-\gamma_x]].\] Thus $\WL$ comes equipped with projections
$\pi_k: \WL \to W_k$ obtained by truncation of a power series, written
as sum of powers of $(\var-\gamma_x)$, in degree $k$.
We have canonical ring homomorphisms $\iota_k: A \to W_k$
given by $\iota_k(x) = \var \bmod {I}{}^k$. They lift to an inclusion
$\iota: A \to \WL$, simultaneously commuting with each $\pi_k$, which
represents elements of $A$ via their $I$-adic expansion.

The \textit{crystalline cohomology} $\hcrys$ of $\phi$ is the
$\WL$-module $ \WL \otimes_{\L[\var]} D(\phi, \L)$, that is, the
completion of $D(\phi,\L)$ at the ideal $I=(\var-\gamma_x)$ of
$\L[\var]$.

Gekeler proved that $D(\phi, \L)$ is a projective, hence free,
$\L[\var]$-module of rank $r$~\cite{Gekeler88}, with canonical basis
${\ehat}{}^{(i)}$ such that ${\ehat}{}^{(i)}(x) = \tau^{i}$ for $1
\leq i \leq r$. From this, it follows that $\hcrys$ is a free
$\WL$-module of rank $r$ as well, as pointed out in~\cite{Angles1997}.

\begin{remark}
  In that reference, $A$ is not necessarily a polynomial ring, and
  $\L[\var]$ is replaced by $A_\L:=\L \otimes_{\F_q} \A$. In this
  case, $D(\phi,\L)$ is a projective $A_\L$-module of rank $r$, the
  definition of ideal $I$ changes, but it remains maximal in
  $A_\L$, so the completion $\WL$ of $A_\L$ at $I$ is still a local
  ring and $\hcrys$ is still free of rank $r$ over $\WL$.
\end{remark}

An endomorphism $u$ of $\phi$ induces an $\L[\var]$-endomorphism $u^*$
of $D(\phi,\L)$, defined as $(u^*(\eta))_x = \eta_x u$, for $\eta$ in
$D(\phi,\L)$; the same holds for the completion $\hcrys$.
Following~\cite{Angles1997}, using the fact that $\hcrys$ is free over
$\WL$, one can then define the characteristic polynomial
$\charpoly_{\WL}(u^*)$ in the usual manner.

Recall now that $\charpoly(u)$ denotes the characteristic polynomial
of $u$, as defined in Section~\ref{ssec:charpoly}. The following
theorem due to Angl\`es~\cite[Thm. 3.2]{Angles1997} relates this
characteristic polynomial to that of the induced endomorphism on
$\hcrys$, where $\iota$ below acts coefficient-wise.

\begin{theorem}\label{maintheorem}  \label{cpoly}
  For $u$ in $\emorph_{\L}(\phi)$, $\charpoly(u)^{\iota} =
  \charpoly_{\WL}(u^*)$.
\end{theorem}

%%%%%%%%%%%%%%%%%%%%%%%%%%%%%%%%%%%%%%%%%%%%%%%%%%%%%%%%%%%%

\subsection{Truncated Cohomology}

Recall now that $\frakp \in A$ is the minimal polynomial of $\gamma_x
\in \L$ over $\F_q$. For $k \ge 1$, we are going to define an
$\F_q$-linear homomorphism $\chi_k$ such that the following diagram
commutes:
\begin{center}
\begin{tikzcd}
 & \WL \arrow[d, "\pi_k"] \\
A \arrow[rd, "\theta_k:f(x) \mapsto f(\var) \bmod \frakp(\var)^k"'] \arrow[r, "\iota_k"] \arrow[ru, "\iota"] & W_k \arrow[d, dashed, "\chi_k"] \\
& \F_q[\var]/(\frakp(\var)^k)
\end{tikzcd}
\end{center}

There exists an isomorphism \[T_k:\F_q[x,\var]/ (\frakp(x),(\var-x)^k)
\to \F_q[y]/(\frakp(y)^k);\] see e.g.~\cite[Lemma~13]{MeSc16}. On the
other hand, recall that $\L=\F_q[t]/(\ell(t))$ is isomorphic
to \[\F_q[x, \vart]/(\frakp(x), g(x, \vart)),\] for some $g$ in
$\F_q[x,\vart]$, monic of degree $n/m$ in $\vart$; in this
representation of $\L$, $\gamma_x$ is simply (the residue class of)
$x$. As a result, we get
\begin{align}
  W_k&= \F_q[t,\var]/(\ell(t),(\var-\gamma_x)^k)\nonumber\\
  &\simeq \F_q[x,\vart,\var]/(\frakp(x),g(x,\vart),(\var-x)^k)\nonumber\\
  &\simeq \F_q[\var,\vart]/(\frakp(\var)^k,G_k(\var,\vart)),\label{eq:Wk}
\end{align}
for a certain polynomial $G_k \in \F_q[\var,\vart]$, monic of degree $n/m$
in $\vart$.

We can then define $\chi_k: W_k \to \F_q[\var]/(\frakp(\var)^k)$ by
\[\chi_k: \sum_{0 \le i < n/m} c_i \vart^i\mapsto c_0,\] and we verify that
it satisfies our claim.  The details of how to compute this
homomorphism are discussed in Section~\ref{sec:algo}.

For $k \ge 1$, we further define the $\textit{precision}$ $k$
cohomology space $\hcrysk$ as the $W_k$-module
\[D(\phi, \L) / I^k\,D(\phi, \L) \simeq \hcrys / I^k\, \hcrys.\]
It is thus free of rank $r$, and an endomorphism $u$ of $\phi$ induces
a $W_k$-linear endomorphism $u^*_k$ of $\hcrysk$.

\begin{remark}
  In~\cite{Gekeler88}, Gekeler introduced {\em de Rham cohomology} of
  Drinfeld modules; this is the case $k=1$ in this construction (in
  which case $W_k=\L$).
\end{remark}

In the following claim,
recall that for a polynomial $P$ and for any map $\chi$ acting on its
coefficient ring, we let $P^{\chi}$ denote coefficient-wise
application of $\chi$ to $P$.

\begin{corollary}\label{maincor}  \label{corpoly}
    For $u$ in $\emorph_{\L}(\phi)$ and $k \ge 1$,
    $\charpoly(u)^{\theta_k} = \charpoly_{W_k}(u^*_k)^{\chi_k}$.
\end{corollary}
\begin{proof}
  Apply ${\chi_k \circ \pi_k}$ coefficient-wise to the equality in
  Theorem~\ref{cpoly}.
\end{proof}

If $u$ has degree $d$ in $\tau$, we know that all coefficients of
$\charpoly(u)$ have degree at most $d$, so they can be recovered from
their reductions modulo $\frakp^k$ for $k = \lceil \frac{d + 1}{m}
\rceil \in O((d+m)/m)$. In the prime field case, where $m=n$, and for
the special case $u=\tau^n$, the above formula gives $k=2$, but we can
take $k=1$ instead; this is discussed in Section~\ref{ssec:other}.

Note also that if we take $k=d+1$, there is no need to consider the
map $\chi_k$: on the representation of $W_{d+1}$ as
\[W_{d+1}=\F_q[x,\vart,\var]/(\frakp(x),g(x,\vart),(\var-x)^{d+1}),\] for $f$ of
degree up to $d$, $\iota_k(f)$ is simply the polynomial $f(\var)$, so
we can recover $f$ from $\iota_k(f)$ for free. We will however refrain
from doing so, as it causes $k$ to increase.

%%%%%%%%%%%%%%%%%%%%%%%%%%%%%%%%%%%%%%%%%%%%%%%%%%%%%%%%%%%%
%%%%%%%%%%%%%%%%%%%%%%%%%%%%%%%%%%%%%%%%%%%%%%%%%%%%%%%%%%%%
%%%%%%%%%%%%%%%%%%%%%%%%%%%%%%%%%%%%%%%%%%%%%%%%%%%%%%%%%%%%

\section{Main Algorithms}\label{sec:algo}

We will now see how the former discussion can be made more concrete,
by rephrasing it in terms of skew polynomials only. The evaluation map
$\eta \mapsto \eta_x$ gives an additive bijection $D(\phi, \L) \to \L\{
\tau \} \tau$. This allows us to transport the $\L[\var]$-module
structure on $D(\phi,\L)$ to $\L\{ \tau \} \tau$: one verifies that it
is given by $(c \var^i * \eta) =c \eta \phi_{x^i}$, for $\eta$ in
$\L\{\tau\}\tau$, $c$ in $\L$ and $i \ge 0$, and that
$\mathcal{B}=(\tau,\dots,\tau^r)$ is a basis of $\L\{\tau\}\tau$ over
$\L[\var]$.

Further, an endomorphism $u \in \emorph_{\L}(\phi)$ now induces an
$\L[\var]$-linear endomorphism $u^\star: \L\{\tau\}\tau \to
\L\{\tau\}\tau$ simply given by $u^\star(v) = v u$ for $v$ in
$\L\{\tau\}\tau$. Reducing modulo the ideal $I^k \subset \L[\var]$,
we denote by $u_k^\star$ the corresponding $W_k$-linear endomorphism
on the quotient module $\L\{\tau\}\tau/I^k_\L \L\{\tau\}\tau\simeq
\hcrysk$.

We can then outline the algorithm referenced in
Theorems~\ref{mainresult} and~\ref{mainresult2}; its correctness
follows directly from Corollary \ref{maincor} and the bound on $k$
given previously.

\begin{enumerate}
\item Set $k = \lceil \frac{d + 1}{m} \rceil$, with $d=\deg_\tau(u)$,
  except if $n=m$ and $u=\tau^n$ (in which case we can take $k=1$)
\item\label{step2} Compute the coefficients $u_{i,1},\dots,u_{i,r} \in
  W_k$ of $\tau^i u \bmod I^k$ on the basis $\mathcal{B}$, for
  $i=1,\dots,r$
\item Using the coefficients computed in step~\ref{step2}, construct
  the matrix for $u^\star_k$ acting on $\L\{\tau\}\tau/I^k_\L
  \L\{\tau\}\tau$ and compute its characteristic polynomial
  $\charpoly_{W_k}(u^\star_k) \in W_k[Z]$
\item Apply the map $\chi_k$ to the coefficients of
  $\charpoly_{W_k}(u^\star_k)$ to recover $\charpoly(u)^{\theta_k}$,
  and thus $\charpoly(u)$.
\end{enumerate}

In Subsections~\ref{ssec:recurrence} to~\ref{ssec:frobenius}, we
discuss how to complete Step~\ref{step2}: we give two solutions for
the case of an arbitrary endomorphism $u$, and a dedicated, more
efficient one, for $u=\tau^n$. We freely use the following notation:
\begin{itemize}
\item for $c$ in $\L$ and $t \in \Z$, let $c^{[t]}$ denote the value
  of the $t$th power Frobenius applied to $c$, that is,
  $c^{[t]}=c^{q^t}$
\item for $f$ in $\L[\var]$, $f^{[t]} \in \L[\var]$ is obtained by
  applying the former operator coefficient-wise, so
  $\deg(f)=\deg(f^{[t]})$
\item for $M=(m_{i,j})_{1 \le i \le u, 1 \le j \le v}$ in
  $\L[\var]^{u\times v}$, $M^{[t]}$ is the matrix with entries
  $(m^{[t]}_{i,j})_{1 \le i \le u, 1 \le j \le v}$.
\end{itemize}
Finally, we define $\mu=(\var-\gamma_x)^k \in \L[\var]$ (with the
value of $k$ defined above); it generates the ideal $I^k$ in $\L[\var]$.

%%%%%%%%%%%%%%%%%%%%%%%%%%%%%%%%%%%%%%%%%%%%%%%%%%%%%%%%%%%%

\subsection{Using a Recurrence Relation}\label{ssec:recurrence}

The following lemma is a generalization of a recurrence noted by
Gekeler~\cite[Section~5]{frobderham} for $r=2$. Recall that we write
$\phi_x = \gamma_x + \Delta_1\tau^1 + \ldots + \Delta_r\tau^r$, with
all $\Delta_i$'s in $\L$; in the expressions below, we write
$\Delta_0=\gamma_x $.

\begin{lemma}\label{lemmarec}
  For any $t \geq 1$, the following relation holds in the $\L[\var]$-module
  $\L\{\tau\}\tau$:
  \begin{equation}\label{mainrec}
    \sum_{i = 0}^{r}\Delta_{i}^{[t]} \tau^{t+i} = \var * \tau^t.
  \end{equation}
\end{lemma}
\begin{proof}
  This follows directly from the module action of $\L[\var]$ on
  $\L\{\tau\}\tau$, by commuting $\tau^t$ across the defining
  coefficients $\Delta_i$ of $\phi$:
$$
  \var * \tau^t = \tau^t \phi_x  = \tau^t \sum_{i = 0}^{r}\Delta_{i} \tau^{i} = \sum_{i = 0}^{r}\Delta_{i}^{[t]} \tau^{t+i}.
\qedhere$$
\end{proof}
For $i=0,\dots,r-1$, let $\Lambda_{i} = -\frac{\Delta_i}{\Delta_r}$
and define the order $t$ companion matrix for the recurrence, $\cA_t
\in\L[\var]^{r\times r}$, as
\begin{equation}
    \cA_t = \begin{bmatrix}
    \Lambda_{r-1}^{[t]} & \Lambda_{r-2}^{[t]} & \ldots & \Lambda_1^{[t]} & \Lambda_0^{[t]} + \frac{\var}{\Delta_r^{[t]}} \\
    1 & 0 & \ldots & 0 & 0 \\
    0 & 1 & \ldots & 0 & 0 \\
    \vdots & \vdots & \ddots & \vdots & \vdots \\
    0 & 0 & \ldots & 1 & 0
    \end{bmatrix}
\end{equation}
For $t \ge 1$, let $\kappa_t \in \L[\var]^{1 \times r}$ denote the coefficient
vector of $\tau^t$ with respect to the standard basis
$\mathcal{B}$. Then, we have the following relation between $r \times
r$ matrices over $\L[\var]$:
\begin{equation}
    \begin{bmatrix}
    \kappa_{t + r} \\
    \kappa_{t +r - 1} \\
    \vdots \\
    \kappa_{t + 1}
    \end{bmatrix}
 = \cA_t \begin{bmatrix}
    \kappa_{t +r- 1} \\
    \kappa_{t +r- 2} \\
    \vdots \\
    \kappa_{t} 
    \end{bmatrix} 
    \end{equation}
For $k\ge 1$, these relations can be taken modulo $\mu$, to give
equalities over $W_k=\L[\var]/\mu$; below, we will write $\bar
\kappa_t =\kappa_t \bmod \mu \in W_k^{1\times r}$.

Starting from $\bar \kappa_t,\dots,\bar \kappa_{t+r-1}$, we obtain $\bar
\kappa_{t+r-1}$ using $O(r)$ operations (divisions, Frobenius) in $\L$
to obtain the coefficients appearing on the first row of $\cA_t$,
followed by $O(kr)$ operations in $\L$ to deduce the entries of $\bar
\kappa_{t+r}$.

Below, we will need $\bar \kappa_1,\dots,\bar \kappa_{d+r}$.  Altogether,
computing them takes $((d+r)krn\log q)^{1 + o(1)}$ bit operations; with
our chosen value of $k$, this is also
\[( (d+r)(d+m)rn/m\log q + )^{1 + o(1)}.\] Let us then write $u = u_0 +
\cdots + u_d \tau^d$. For $i=1,\dots,r$, we have
\[\tau^i u = u_0^{[i]}\tau_i +\cdots + u_d^{[i]} \tau^{d+i},\]
so the coefficient vector
$[u_{i,1} \cdots u_{i,r}] \in W_k$ of $\tau^i u \bmod I^k$ on the basis
$\mathcal{B}$ is given by the product
\[[u_0^{[i]}~\cdots~u_d^{[i]}]
\begin{bmatrix}
\bar \kappa_{i} \\
\bar\kappa_{i+1} \\
    \vdots \\
~\bar \kappa_{i+d}~
\end{bmatrix} \in W_k^{1 \times r}.\]
Each such operation takes $O(dkrn)$ operations in $\L$, for a total of
$(d(d+m)r^2n/m \log q)^{1+o(1)}$ bit operations if done independently
of one another (this is the dominant cost in the algorithm).

In cases when $d$ is not small compared to $r$, we can reduce the cost
slightly using matrix arithmetic, since all coefficient vectors we
want can be read off an $r \times (d+r) \times r$ matrix product,
$$
\begin{bmatrix}
  u_0^{[1]} & \cdots & u_d^{[1]} & 0 &\cdots&\cdots &0\\
  0 & u_0^{[2]} & \cdots & u_d^{[1]} & 0& \cdots &0 \\
   && \ddots&&\ddots \\
  0 &\cdots&\cdots &0 & u_0^{[r]} & \cdots & u_d^{[r]}
\end{bmatrix}
\begin{bmatrix}
\bar \kappa_{1} \\
\bar\kappa_{i+1} \\
    \vdots \\
~\bar \kappa_{d+r}~
\end{bmatrix} \in W_k^{r \times r}.$$
This takes 
$((d+r)(d+m)r^{\omega-1}n/m \log q)^{1+o(1)}$ bit operations.

%%%%%%%%%%%%%%%%%%%%%%%%%%%%%%%%%%%%%%%%%%%%%%%%%%%%%%%%%%%%

\subsection{Using Euclidean Division}\label{ssec:division}

This section describes an alternative approach to computing the
coefficients of an endomorphism $u$ on the canonical basis
$\mathcal{B}$. Computations are done in $\L[\var]$ rather than
$W_k=\L[\var]/\mu$ (we are not able to take reduction modulo $\mu$
into account in the main recursive process).

The algorithm is inspired by a well-known analogue for commutative
polynomials~\cite[Section~9.2]{GaGe13}: for a fixed $a \in \L[\var]$
of degree $r$, we can rewrite any $f$ in $\L[\var]$ as $f=\sum_{0 \le
  i < r} f_i(a) \var^i$, for some coefficients $f_0,\dots,f_{r-1}$ in
$\L[\var]$. This is done in a divide-and-conquer manner.

This approach carries over to the non-commutative setting. We start by
showing how $f$ of degree $d$ in $\L\{\tau\}$ can be rewritten as
$$f = \sum_i f_i \phi_x^i,$$ for some $f_i$ of degree less than $r$ in
$\L\{\tau\}$. If we let $K$ be such that $d < K r \le 2d$, with $K$ a
power of $2$, index $i$ in the sum above ranges from $0$ to $K-1$.

If $K=1$, we are done. Else set $K'=K/2$, and compute the quotient $g$
and remainder $h$ in the right Euclidean division of $f$ by
$\phi_x^{K'}$, so that $f = g\phi_x^{K'} + h$. Recursively, we compute
$g_0\dots,g_{K'-1}$ and $h_{0},\dots,h_{K'-1}$, such that
\[g = \sum_{0 \le i < K'} g_{i} \phi_x^i \quad\text{and}\quad h =
\sum_{0 \le i < K'} h_i \phi_x^i.\] Then, we return
$h_0,\dots,h_{K'-1},g_0,\dots,g_{K'-1}$. The runtime of the whole
procedure is $\softO({\sf SM}(d,n,q))$ bit operations, with $\sf SM$
as defined in Section~\ref{sec:prelim} (the analysis is the same as
the one done in the commutative case in~\cite{GaGe13}, and uses the
super-linearity of {\sf SM} with respect to $d$).

From there, we are able to compute the coefficients of $f \in
\L\{\tau\}\tau$ on the monomial basis $\mathcal{B}$. This essentially
boils down to using the procedure above, taking care of the fact that
$f$ is a multiple of $\tau$. Factor $\tau$ on the left, writing $f$ as
$\tau g$: if $f=F \tau$, $g=F^{[-1]}$. Apply the previous procedure,
to write $g=\sum_{0 \le i \le s} g_i \phi_x^i$, with all $g_i'$ of
degree less than $r$ and $s \le d/r$.

This gives $f = \tau g = \sum_{0 \le i \le s} (g^{[1]}_i \tau)
\phi_x^i$, with all coefficients $g^{[1]}_i \tau$ supported on
$\tau,\dots,\tau^r$. Extracting coefficients of $\tau,\dots,\tau^r$,
we obtain polynomials $G_1,\dots,G_r$ of degree at most $s$ in
$\L[\tau]$ such that $f = \sum_{1\le i \le r} G_i * \tau^i$.

The cost of left-factoring $\tau$ in $f$, and of multiplying all
coefficients of $g$ back by $\tau$, is $(d n \log q)^{1+o(1)}$, so the
dominant cost is $\softO({\sf SM}(d,n,q))$ bit operations from the
divide-and-conquer process. To obtain the matrix of an endomorphism
$u$ of degree $d$, we apply $r$ times this operation, to the terms
$\tau^i u$, $i=1,\dots,r$. The runtime is then dominated by $\softO(r
{\sf SM}(d+r,n,q))$. Finally, reducing the entries of the matrix
modulo $\mu=(\var-\gamma_x)^k$ takes softly linear time in the size of
these entries, so can be neglected.

%%%%%%%%%%%%%%%%%%%%%%%%%%%%%%%%%%%%%%%%%%%%%%%%%%%%%%%%%%%%

\subsection{Special Case of the Frobenius Endomorphism}\label{ssec:frobenius}

In the particular case where $u = \tau^n$, we may speed up the
computation using a baby-step giant-step procedure, based on the
approach used in~\cite{DOLISKANI2021199}. As a first remark, note that
for $u=\tau^n$, $d=n$ and $k$ in $O(n/m)$.

In this case, it is enough to compute the vectors
$\bar\kappa_{n+1},\dots,\bar\kappa_{n+r}$. They are given by
\begin{equation}
    \begin{bmatrix}
    \bar\kappa_{n+r} \\ \bar\kappa_{n +r - 1} \\ \vdots \\ \bar\kappa_{n +
      1}
    \end{bmatrix} = \bar\cA_n\hdots \bar\cA_1,
    \end{equation}
with $\bar \cA_t$ the image of $\cA_t$ modulo $\mu=(\var-\gamma_x)^k$
for all $t$.  To compute the matrix product $\bar \cA = \bar \cA_n
\hdots\bar \cA_1$, we slightly extend the approach used
in~\cite{DOLISKANI2021199} (which dealt with the case $k=1$). Consider
the following element of $\L[\var]^{r \times r}$:
\begin{equation}\label{eqdef:B}
    \cB =
    \begin{bmatrix}
    \Lambda_{r-1} & \Lambda_{r-2} & \ldots & \Lambda_1 & \Lambda_0 \\
    1 & 0 & \ldots & 0 & 0 \\
    0 & 1 & \ldots & 0 & 0 \\
    \vdots & \vdots & \ddots & \vdots & \vdots \\
    0 & 0 & \ldots & 1 & 0
    \end{bmatrix} + \begin{bmatrix}
    0 & 0 & \ldots & \Delta_r^{-1} \\
    0 & 0 & \ldots & 0 \\
    \vdots & \vdots & \ddots & \vdots\\
    0 & 0 & \ldots & 0
    \end{bmatrix} \var.
\end{equation}
It follows in particular that for $t \ge 1$,
$$\cA_t = \cB^{[t]} \quad\text{and}\quad \bar\cA_t = \cB^{[t]} \bmod
\mu,$$ with reduction applied coefficient-wise.

Write $n^* = \lceil \sqrt{nk} \rceil \in O(n/\sqrt m)$,
and let $n$ be written as $n = n^*n_1 + n_0 $ with
$0 \le n_0 < n^*$, so that $n_1 \le \sqrt{n/k}$.
Setting
$$\cC = \cB^{[n^* + n_0]} \cdots \cB^{[n_0 + 1]}$$ and
$$\cC_0 =\cB^{[n_0]} \ldots \cB^{[1]},$$
the matrix $\cA$ is the product
\begin{equation*}
  \cA = \cC^{[(n_1-1) n^*]} \cdots \cC^{[n^*]} \cC \cC_0.
\end{equation*}
Our goal is to compute $\bar \cA=\cA \bmod \mu$, without computing
$\cA$ itself.

Any Frobenius application (of positive or negative index) in $\L$
takes $(n \log q)^{1+o(1)}$ bit operations. In particular, computing
all matrices $\cB^{[i]}$ that arise in the definitions of $\cC$ and
$\cC_0$ takes $(r n^2/\sqrt{m} \log q)^{1+o(1)}$ bit operations.

Once they are known, the next stage of the algorithm computes $\cC$
and $\cC_0$ in $\L[\var]$. This is done using a matrix subproduct-tree
algorithm~\cite[Chapter~10]{GaGe13}, using a number of operations in
$\L$ softly linear in $r^\omega n^*$. This is $(r^\omega n^2/\sqrt{m}
\log q)^{1+o(1)}$ bit operations.

To deduce the shifted matrices
$$\cC^{[(n_1-1)n^*]} \bmod \mu,\dots,\cC^{[n^*]}\bmod\mu,$$ we use the
following lemma.
\begin{lemma}
  For $f$ in $\L[\var]$ and $t \ge 0$,
  $$f^{[t]} \bmod \mu = (f \bmod \mu^{[-t]})^{[t]}$$
\end{lemma}
\begin{proof}
  Let $g=f \bmod \mu^{[-t]}$, so that we have an equality of the form
  $f = a \mu^{[-t]} + g$ in $\L[\var]$. We raise this to the power
  $q^t$ coefficient-wise; this gives $f^{[t]} = a^{[t]}\mu+
  g^{[t]}$. Since $g$, and thus $g^{[t]}$, have degree less than $k$,
  this shows that $g^{[t]}=f^{[t]} \bmod \mu$.
\end{proof}

Applying this entry-wise, we compute $\cC^{[i n^*]}\bmod \mu$ by
reducing all entries of $\cC$ modulo $\mu^{[-i n^*]}$, then raising
all coefficients in the result to the power $q^{i n^*}$, for
$i=1,\dots,(n_1-1)$.

Matrix $\cC$ has degree $O(n/\sqrt{m})$, and the sum of the degrees of
the moduli $\mu^{[-t]}$ is $kn_1$, which is $O(n/\sqrt{m})$ as
well. Altogether, this takes $O(r^2n/\sqrt{m})$ applications of
Frobenius in $\L$, together with $O(r^2n/\sqrt{m})$ arithmetic
operations in $\L$ to perform all Euclidean
divisions~\cite[Chapter~10]{GaGe13}. Thus, the runtime is $(r^2
n^2/\sqrt{m} \log q)^{1+o(1)}$ bit operations.

Finally, we multiply all matrices $\cC^{[i n^*]}\bmod \mu$ and $\cC_0
\bmod \mu$. This takes $(r^\omega n^2/\sqrt{m} \log q)^{1+o(1)}$ bit
operations.

%\resizebox{\columnwidth}{!}{%

\begin{algorithm}[!t]
  %% \caption{Characteristic Polynomial of the Frobenius}
  \label{euclid}
\begin{algorithmic}[1]
  \Procedure{CharPolyFrobenius}{} \\ \textbf{Input} A field extension
  $\L$ of degree $n$ over $\F_q$, $(\Delta_1, \ldots, \Delta_r) \in
  \L^r$ representing a rank $r$ Drinfeld module $\phi$ over $(\L,
  \gamma)$.\\ \textbf{Output} $a_i \in A$ such that the characteristic
  polynomial of the Frobenius is $ X^r + \sum_{i=0}^{r-1} a_i X^i$.
  \State $n^*, n_1, n_0 \gets \lceil \sqrt{nk} \rceil, \lfloor n / n^* \rfloor, n \bmod n^* $.
  \State $\cB$ as in~\eqref{eqdef:B}
  \State $\cC \gets \cB^{[n^* + n_0]} \ldots \cB^{[n_0 + 1]}$.
  \State $\bar \cC_0 \gets \cB^{[n_0]} \ldots \cB^{[1]} \bmod \mu$
  \State $\bar \cC^{[in^*]}\gets (\cC \bmod \mu^{[-in^*]})^{[in^*]}$ for $0 \leq i < n_1$.
  \State $\bar \cA \gets \bigg(\displaystyle\prod_{i=0}^{n_1 - 1}\bar \cC^{[in^*]} \bigg) \bar \cC_0$
  \State $\bar a_i \gets \textnormal{coefficient of } Z^i \textnormal{ in } \det(\bar \cA - ZI)$
 \State \Return $a_i = \chi_k(\bar a_i)$ for $0 \leq i < r$
  
  %\State \Return $( \chi_k(\bar a_0), \chi_k(\bar a_1), \ldots, \chi_k(\bar a_{r-1}) )$
\EndProcedure
\end{algorithmic}
\end{algorithm}

%%%%%%%%%%%%%%%%%%%%%%%%%%%%%%%%%%%%%%%%%%%%%%%%%%%%%%%%%%%%

\subsection{Other Operations}\label{ssec:other}

Once the coefficients of the skew polynomials $\tau^i u$ on the basis
$\mathcal B$ are known modulo $\mu$, we compute the characteristic
polynomial of the matrix formed from these coefficients. This can be
done with a bit cost of $(r^{\lambda} kn\log q)^{1+o(1)}$ when the
matrix has entries in $W_k$, with $\lambda$ the exponent
defined in Section~\ref{ssec:charpoly}.

At this stage, we have all coefficients of
$\charpoly_{W_k}(u^\star_k)$ in $W_k$. It remains to apply the map
$\chi_k$ to each of them to recover $\charpoly(u)$.

Elements of $W_k=\F_q[t,\var]/(\ell(t),(\var-\gamma_x)^k)$ are written
as bivariate polynomials in $t,\var$, with degree less than $n$ in $t$
and less than $k$ in $\var$. To compute their image through $\chi_k$,
we first apply the isomorphisms
\begin{align*}
  W_k= \F_q[t,\var]/(\ell(t),(\var-\gamma_x)^k)\nonumber&\xrightarrow{A_k} \F_q[x,\vart,\var]/(\frakp(x),g(x,\vart),(\var-x)^k)\nonumber\\
  &\xrightarrow{B_k} \F_q[\var,\vart]/(\frakp(\var)^k,G_k(\var,\vart))
\end{align*}
from~\eqref{eq:Wk}, with $\frakp(\var)^k$ of degree $km$ and $G_k$ of
degree $n/m$ in $t$.

We mentioned in Section \ref{sec:prelim} that for $c$ in
$\L=\F_q[t]/(\ell(t))$, we can compute its image $\alpha(c)$ in
$\F_q[x,\vart]/(\frakp(x),g(x,\vart)$ using $(n \log q)^{1+o(1)}$ bit
operations. Proceedings coefficient-wise with respect to $\var$, this
shows that for $C$ in $W_k$, we can compute $A_k(C)$ in $(k n \log
q)^{1+o(1)}$ bit operations.

The \textit{tangling} map of ~\cite[\S 4.5]{powermod} provides an
algorithm for computing the isomorphism $\F_q[x, y]/(\frakp(x),
(\var-x)^k) \to \F_q[\var]/(\frakp(\var)^k)$ in $(km \log q)^{1+o(1)}$
bit operations (this could also be done through modular composition,
with a similar asymptotic runtime, but the algorithm
in~\cite{powermod} is simpler and faster). Applying it
coefficient-wise with respect to $t$, this allows us to compute
$B_k(A_k(C))$ in $(k n \log q)^{1+o(1)}$ bit operations again.  At
this stage, the mapping $\chi_k$ is simply extraction of the degree-0
coefficient in $t$.

We apply this procedure $r$ times, for a total cost of $(r k n \log
q)^{1+o(1)}$ bit operations. This can be neglected in the runtime
analysis.

When using precision $k = 1$ for the prime field case, for $u=\tau^n$,
it is necessary to compute the constant coefficient $a_0$ separately.
This is done using the formula $a_0 = (-1)^{n(r+1) +
  r}N_{\L/\F_q}(\gamma_{\Delta_r})^{-1} \frakp$ from \cite{garaipap}
and takes $(n \log q)^{1 + o(1)}$ bit operations.

Summing the costs seen so far for the various steps of the algorithm
finishes the proof of our main theorems.

%%%%%%%%%%%%%%%%%%%%%%%%%%%%%%%%%%%%%%%%%%%%%%%%%%%%%%%%%%%%

\subsection{Example}

  Let $\F_q = \Z/2\Z$, $n = 3$ and set $\ell(t) = t^3 + t + 1$ and $\L
  = \F_2[t]/(\ell(t))$. Let $\gamma_x = t + 1 \bmod \ell(t)$, so that
  \[\frakp = x^3 + x^2 + 1 = \ell(x+1),\]
  and $\L \cong \F_{\frakp}=\F_q[x]/(\frakp(x))$, with the isomorphism
  given by $f(t) \mapsto f(x+1)$.
  Consider the rank 4 Drinfeld module $\phi_x = t \tau^4 + (t^2 +
  t)\tau^3 + \tau^2 + t^2 \tau + t + 1$. We proceed to compute the
  characteristic polynomial using the de Rham cohomology, that is,
  crystalline cohomology truncated in degree $k=1$. In other words,
  all computations are done over $\L$

  The recurrence of equation (\ref{mainrec}) becomes $ \tau^{k + 4} =
  (t + 1)^{2^k}\tau^{(k + 3)} + (t^2 + 1)^{2^k}\tau^{k + 2} +
  t^{2^k}\tau^{k + 1} + (1 + t^{1 -2^k})\tau^{k} $. Running the
  recurrence for $n = 3$ iterations gives:
  \begin{itemize}
  \item $\tau^{5} = (t^2 + 1)\tau^{4} +  (t^2 + t + 1)\tau^{3} + t^2 \tau^{2} + t^2 \tau^{1}$
  \item $\tau^{6} = (t^2 + 1)\tau^{4} +  (t^2 + 1)\tau^{3} + (t^2 + t) \tau^{2} + \tau^{1}$
  \item $\tau^{7} = \tau^{4} +  t\tau^{3} + (t + 1) \tau^{2} + \tau^{1}$
\end{itemize}
  A matrix for the Frobenius endomorphism 
  can be inferred to be
  \begin{center}
$ \begin{bmatrix}
          1     &      t   &    t + 1  &         1 \\
    t^2 + 1  &   t^2 + 1   &  t^2 + t &          1 \\
    t^2 + 1 & t^2 + t + 1    &     t^2    &     t^2 \\
          1      &     0     &      0    &       0 \\
\end{bmatrix} . $

\end{center}
It has characteristic polynomial $Z^4 + (t + 1)Z^2 + (t + 1)Z$. Using
the expression for $a_0$ which is valid in the prime field case, the
Frobenius norm can be inferred to be $a_0 = x^3 + x^2 + 1$.

To recover the final coefficients, observe that $t \mapsto x + 1$
gives the required map $\chi_1 : W_1 = \L \to
\F_{\frakp}$. Finally, we conclude that the characteristic polynomial
of $\tau^n$ is $Z^4 + xZ^2 + xZ + x^3 + x^2 + 1$.

%%%%%%%%%%%%%%%%%%%%%%%%%%%%%%%%%%%%%%%%%%%%%%%%%%%%%%%%%%%%
%%%%%%%%%%%%%%%%%%%%%%%%%%%%%%%%%%%%%%%%%%%%%%%%%%%%%%%%%%%%
%%%%%%%%%%%%%%%%%%%%%%%%%%%%%%%%%%%%%%%%%%%%%%%%%%%%%%%%%%%%

\section{Experimental Results}

An implementation of the algorithm of section (\ref{ssec:frobenius}) was created in
SageMath~\cite{sagemath} and is publicly available at
\url{https://github.com/ymusleh/drinfeld-module}. An implementation in
MAGMA~\cite{MR1484478} is also publicly available at
\url{https://github.com/ymusleh/drinfeld-magma} and was used to
generate the experimental results included in this work. Our
implementation differs from our theoretical version in a few ways.
\begin{itemize}
\item The Kedlaya-Umans algorithm is most likely not used by MAGMA
  for computing Frobenius mappings of elements of~$\L$.
\item To compute the images of coefficients under the map $\chi_k$, we
  leverage a simpler procedure using reduction modulo bivariate
  Gr\"obner bases, rather than the tangling map of van der Hoeven and
  Lecerf. In any case, this does not impact the run times presented.
\end{itemize}

\begin{acks}
  We thank Xavier Caruso, Antoine Leudi\`ere and Pierre-Jean
  Spaenlehauer for interesting discussions. Schost is supported by an
  NSERC Discovery Grant.
\end{acks}

\begin{table}[h!]
\resizebox{\columnwidth}{!}{%
\begin{tabular}{|llllllll|}
\hline
\multicolumn{8}{|c|}{Run Times for $m = 10$  $q = 25$ in seconds}                                                                                                                                                                                                                                                       \\ \hline
\multicolumn{1}{|l|}{\textbf{}}       & \multicolumn{1}{l|}{$n = 100$} & \multicolumn{1}{l|}{$n = 150$} & \multicolumn{1}{l|}{$n = 200$} & \multicolumn{1}{l|}{$n = 300$} & \multicolumn{1}{l|}{$n = 400$} & \multicolumn{1}{l|}{$n = 500$} & $n = 600$ \\ \hline
\multicolumn{1}{|l|}{$r = 5$}  & \multicolumn{1}{l|}{0.400}                & \multicolumn{1}{l|}{2.260}                 & \multicolumn{1}{l|}{42.190}                 & \multicolumn{1}{l|}{86.830}                 & \multicolumn{1}{l|}{269.760}                 & \multicolumn{1}{l|}{635.170}                 & 1099.110               \\ \hline
\multicolumn{1}{|l|}{$r = 9$}  & \multicolumn{1}{l|}{0.790}                 & \multicolumn{1}{l|}{4.210}                 & \multicolumn{1}{l|}{78.860}                 & \multicolumn{1}{l|}{157.100}                 & \multicolumn{1}{l|}{481.090}                 & \multicolumn{1}{l|}{1129.670}                 &                  \\ \hline
\multicolumn{1}{|l|}{$r = 12$} & \multicolumn{1}{l|}{1.170}                 & \multicolumn{1}{l|}{6.080}                 & \multicolumn{1}{l|}{104.630}                 & \multicolumn{1}{l|}{220.430}                 & \multicolumn{1}{l|}{658.950}                 & \multicolumn{1}{l|}{1531.580}                 &                  \\ \hline
\multicolumn{1}{|l|}{$r = 18$} & \multicolumn{1}{l|}{2.300}                 & \multicolumn{1}{l|}{11.360}                 & \multicolumn{1}{l|}{170.790}                 & \multicolumn{1}{l|}{366.690}                 & \multicolumn{1}{l|}{1074.840}                 & \multicolumn{1}{l|}{2451.530}                 &                  \\ \hline
\multicolumn{1}{|l|}{$r = 23$} & \multicolumn{1}{l|}{3.820}                 & \multicolumn{1}{l|}{17.580}                 & \multicolumn{1}{l|}{240.100}                 & \multicolumn{1}{l|}{525.670}                 & \multicolumn{1}{l|}{1518.370}                 & \multicolumn{1}{l|}{}                 &              \\ \hline
\end{tabular}
}
\end{table}

\FloatBarrier
%%%%%%%%%%%%%%%%%%%%%%%%%%%%%%%%%%%%%%%%%%%%%%%%%%%%%%%%%%%%%%%%%%%%%%%%%
%%%%%%%%%%%%%%%%%%%%%%%%%%%%%%%%%%%%%%%%%%%%%%%%%%%%%%%%%%%%%%%%%%%%%%%%

\bibliographystyle{ACM-Reference-Format}
\bibliography{drinfeld}
%\printbibliography
\end{document}